\documentclass{article}
\usepackage[utf8]{inputenc}
\usepackage{amsmath}
\usepackage{amsfonts}
\usepackage{amssymb}
\usepackage{makeidx}
\usepackage{graphicx}
\usepackage{xcolor}
 \usepackage[intlimits]{esint}
\usepackage{tikz}
\newtheorem{theorem}{Theorem}[section]

\newtheorem{remark}{Remark}
\newtheorem{example}{Example}

\def \Box {\vrule height5pt width5pt depth0pt}
\def \endproof{\hfill \Box \vskip .5cm}

\title{A note on averaging prediction accuracy, Green's functions and other  kernels}

\author{  Juan Galvis\thanks{Departamento de Matem\'aticas, Universidad Nacional de Colombia, Carrera 45 No. 26-85, Edificio Uriel Gutierr\'ez, Bogot\'a D.C., Colombia, \tt \{jcgalvisa, fohernandezr, fagomezj\}@unal.edu.co.} \and
Freddy Hern\'andez-Romero$^*$
\and Francisco G\'omez$^*$}

\begin{document}

\maketitle

%\tableofcontents

\begin{abstract} We present the mathematical context of the predictive accuracy index and then introduce the definition of integral average transform. We establish the relation of our definition with two variables kernels $K({\bf y},{\bf x})$.  As an example of an application we show that integrating against the fundamental solution of the Laplace operator, that is, solving the Poisson equation, can be re-interpreted as an integral of averages of the forcing term over balls. As a result, we obtained a  novel integral representation of the solution of the Poisson equation. Our motivation comes from the need for a better mathematical understanding of the prediction accuracy index. This index is used to identify hot spots in predictive security and other applications.\\

\end{abstract}

\section{Introduction}

We start by analyzing a procedure that is used in
predictive security applications  to
 identify crime hot spots, \cite{Chainey_2008,joshi2021considerations,drawve2016metric, mohler2020learning}. This procedure is known as the Prediction Accuracy Index (PAI) and it is used to compare among possible hot spots regions. In this context, a higher PAI is preferred to a lower one. 
 In particular, in some applications it is used (along side other measures such as distances and divergences \cite{mohler2020learning,drawve2016metric}) as an indicator of similarity between two densities of random variables (a density from a prediction and the other from an observation). 
 Roughly speaking, the PAI of a region is defined as the ratio between the hit rate and the volume proportion (see \cite{Chainey_2008,joshi2021considerations,drawve2016metric} and the next section).
When the PAI is used as a similarity measure between two densities, say $\phi$  and $\psi$, this indicator is computed as follows
\begin{enumerate}
    \item Identify several regions of high probability according to $\psi$.
    \item Evaluate the PAI on each one of the regions identified before.  Here the idea is that higher values of PAI will indicate certain similarity between  $\phi$ and $\psi$.
\end{enumerate}
In case several regions of high-probability are 
selected in 1. and 2. above, one can add to this procedure the following:
\begin{enumerate}
\item[3.] Compute the average among the PAI indicators obtained above. See Section \ref{sec:motivation}.
\end{enumerate}

Our preliminary analysis readily shows that PAI measures (even average ones) shall  not be used as a similarity measure between two densities and it have to be used cautiously to identify hot spots. See also 
\cite{joshi2021considerations,drawve2016metric,mohler2020learning} for some other issues related to the PAI indicator. 
We believe this is an important take home message in this note. See
Remark \ref{remarkPAI}.\\

Anyhow, the procedure described above (1., 2. and 3.)  led us to the 
 definition of the integral average transform. 
See Section \ref{sec:IAT} and Remark \ref{remarkPAIandIAT}. The integral average transform is defined for a given function 
$f({\bf y})$, assumed regular enough just to fix ideas. The transform is the result of integrating {\it plain unweighted} averages of $f$ over a family of neighborhoods of 
a point ${\bf x}$, say $B_{s,{\bf x}}$ that is indexed by the real parameter $s$. 
Note that we are talking about plain averages 
$\frac{1}{|B_{s,{\bf x} }|}\int_{B_{s,{\bf x} }} f$ and not weighted averages such as 
$\frac{1}{|B_{s,{\bf x} }|}\int_{B_{s,{\bf x} }} \omega f$.
See 
\eqref{eq:defIntegralAverageTransform} for a precise definition. We show that this procedure is equivalent to computing integrals against a two variables kernel $K({\bf y},{\bf x})$ with a possible singularity at $x$, that is, computing, 
\[\int f({\bf y})K({\bf y},{\bf x})d{\bf y}.
\] 
We also show that the integral above can be interpreted as  an integral average transform, that is, an integral of {\it regular} averages of $f$ over a parametrized family of regions.
As an example of application we show that integrating against the
fundamental solution of the Laplace equation, that is, solving the Poisson equation, can be re-interpreted as an integral of averages of the forcing terms over balls. See Section \ref{sec:IAT} and Remark \ref{remarkGreenIAT}.

The rest of the paper is organized as follows. 
In Section \ref{sec:motivation} we present the mathematical context related to the PAI including some modification know as the penalized PAI. We introduce the average PAI that motivates our main definition. In Section \ref{sec:IAT} we present the definition of integral average transform and, as an application of this novel interpretation of kernels, we show how to write the fundamental solution of the Laplace operator as an integral average transform. 
In Section \ref{sec:conclusions} we present some final comments.

\section{Averaging prediction accuracy}\label{sec:motivation}

In this section we present the motivation that lead us to study integrals of average of function on parametrized regions.\\

Let us consider the real random variables $X$ and $Y$. Denote by 
$\phi=d\mu_X$ and $\psi=d\mu_Y$ be probability density functions associated with $X$ and $Y$, respectively.  Given a measurable set $B$ we write $|B|$ for the Lebesgue measure of $B$.
Denote by $[\psi> r ]$ the set with probability density level higher (greater) than $r\in \mathbb{R}$, that is,
\begin{equation}
[\psi> r ]=\{ {\bf z}\in \mathbb{R}^n  :\quad \psi({\bf z}) > r \} = \psi^{-1}((r,+\infty)). 
\end{equation}
% For $r$ close to $R:=\max_{\bf z}\psi({\bf z})$, the set $[\psi> r]\subset \mathbb{R}^n$ is the region with most probable occurrences of $Y$ and in some practical applications is refereed to as \emph{hotspot area} (and in case of being the union of disjoint small regions these regions are know as a hotspots). 

Let us define $R:=\max_{\bf z}\psi({\bf z})$. In some practical applications a threshold value $r \in (0,R)$ is determined and the region $[\psi> r]\subset \mathbb{R}^n$ is refereed to as \emph{hot spot area} (and in case of being the union of disjoint small regions these regions are know as hot spots), see \cite{Chainey_2008} and references therein. In fact, note that if $r$ is close to $R$ the set $[\psi> r]\subset \mathbb{R}^n$ will contain the most probable occurrences of $Y$.

It is common in applications to use the hot spots of the random variable $Y$ to predict the hot spots of another a random variable $X$. A common way to measure the performance of such predictions is by using the so-called
\emph{Prediction Accuracy Index} (PAI) defined next (see \cite{Chainey_2008}).
In some cases, when several possible predictions of $X$ are available, it is used the PAI measure to decide which model to use to report hots spots predictions of the target $X$. See \cite{Chainey_2008,joshi2021considerations}.

In order to fix ideas and simplify the presentation we assume that $\phi$ and $\psi$ are continuous functions but we can as well consider discretized version up to some given spatial resolution as it is most common in applications.

\subsection{PAI of a subregion}
Assume that we want to analyze a measurable study region $A$ with positive volume measure. Region $A$  corresponds to the study region where 
the random variable $X$  occur. 
The region $B\subset A$ is a possible hot spot region for $X$ in $A$. Define the hit rate by (\cite{Chainey_2008}) 
\[
H_{\phi}(B)=\frac{\displaystyle \int_{B} \phi}{\displaystyle \int_A \phi}.
\]
If $|B|>0$ define the average of $\phi$ on $B$ by 
\[
 \fint_{B} \phi=\frac{1}{|B|} \int_{B} \phi.
\]
Note that in our notation for $H_\phi$ we do not make explicit reference to the region $A$ since we assume $A$ is fixed.
A commonly used indicator (\cite{Chainey_2008}) of the quality of using $B$ to approximate the hot spots of $X$ in $A$ is given by the ratio of the hit rate with the percentage of volume $\frac{|B|}{|A|}$, that is, 
\[
PAI_\phi(B)= \frac{H_\phi(B)}{\frac{|B|}{|A|}}
=\frac{|A|}{\displaystyle \int_{A}\phi }\frac{\displaystyle \int_{B}\phi }{|B|}
=\frac{\displaystyle \fint_{B} \phi}{\displaystyle \fint_{A} \phi}.\]
Note that the $PAI_\phi(B)$ can also be  written as the ratio between the averages of $\phi$ on 
$B$ and $A$ which gives another interpretation of the $PAI$ that readily reveals some possible issues with this indicator such as the preferences for small subregions with high-value of $\phi$.\\

In order to better understand this indicator we present some simple examples.
Note that if  $B=A$ we get $PAI_\phi(B)=1$. If for instance we select $B$ with volume percentage
\[
\frac{|B|}{|A|}=0.5
\]
and  hit rate 
\[
H(B)=0.5
\]
we still have $PAI_\phi(B)=1$ since we can predict half of occurrences of $X$ in half of the study area.
Assume now that we have 
\[
\frac{|B|}{|A|}=0.2
\mbox{ and }
H(B)=0.8.
\]
We then have $PAI_\phi(B)=4$ since we can predict 80\% of occurrences of $X$ in 20\%  of the study area.

In practical applications we are interested in finding regions with a high value of PAI for a given density 
$\phi$. We make the following observations:

\begin{itemize}
\item Since we always have 
$ \fint_{B} \phi\leq R=\max_{{\bf z}\in\mathbb{R}^n} \phi({\bf z})$ we see that
\[
PAI_\phi(B)
=\frac{\displaystyle \fint_{B} \phi}{\displaystyle \fint_{A} \phi}
\leq \frac{R}{\displaystyle \fint_{A} \phi}.
\]

\item Let $\phi^{-1}(R)=\{ {\bf z}\in \mathbb{R}^n \, : \, 
\phi({\bf z})=R\}=\arg\max_{\bf z}\psi({\bf z})$. If  $B^*=\phi^{-1}(R)\cap A$ is such that  $|B^*|>0$ then $B^*$ 
 and all its possitive measure subsets maximizes the PAI indicator. In fact, 
 for all $\widetilde{B}\subset B^*$ such that  $|\widetilde{B}|>0$ we have 
\[
PAI_\phi(\widetilde{B})
= \frac{R}{\displaystyle \fint_{A} \phi}.
\]
We then see that, in practical applications the PAI indicator will favor small area regions around the maximum values of the density $\phi$. Here small will depend on the spatial resolution at which subregions are computed. This might not be convenient in some applications as pointed out in \cite{joshi2021considerations}.

\item If $B^*$ defined above is such that 
$|B^*|=0$ then, given a subregion $\widetilde{B}\subset A$ with 
$|\widetilde{B}|>0$ and 
$B^*\subsetneq \widetilde B$, there always exists 
$\widehat{B}$ with $|\widehat{B}|>0$ and  $ B^*\subsetneq \widehat B \subsetneq \widetilde B$  such that 
\[
PAI_\phi(\widetilde B)
< PAI_\phi(\widehat B) =\frac{\displaystyle \fint_{\widehat B} \phi}{\displaystyle \fint_{A} \phi}
<\frac{R}{\displaystyle \fint_{A} \phi}.
\]
\end{itemize}

Therefore we see that searching for Borel subregions 
with as high PAI as possible is not a well posed problem under general considerations. Other possible optimization problems may be needed for practical applications. See \cite{joshi2021considerations, drawve2016metric,mohler2020learning} for more details.\\

Due to the fact that the PAI indicator prefers small area regions 
some modifications have been introduced. Among them, in  
\cite{joshi2021considerations} was introduced a penalized PAI where 
the area ratio is penalized, that is, 
\[
PPAI_\phi(B)= \frac{H_\phi(B)}{\left(\frac{|B|}{|A|}\right)^\alpha}
= 
\left(\frac{|B|}{|A|}\right)^{1-\alpha}
PAI_\phi(B)=\lambda(B)PAI(B).
\]
Here there was introduced the penalization 
$\lambda(B)$. The exponent $\alpha$ may depend on $B$, e.g.,   $\alpha=H_\phi(B)$ (and in this case, for 
small hit rate, the indicator is penalized multiplicatively by the volume proportion of the region $B$ while for large hit rate close to 1 we do not have that penalization, \cite{joshi2021considerations}).
Many other penalization alternatives can also be consider at the light of practical applications. For instance, a penalization of the form  
\[
\lambda(B)=\frac{|B|}{|\partial B|}
\]
where $|\partial B|$ is the surface volume of $\partial B$ with $B$ regular enough. See \eqref{weight} below. \\

One additional observation is the following. 
One can think to compute an average of several PAI values over subregions. 
Let \[B_{N}\subseteq B_{N-1}\subseteq \cdots \subseteq B_{1}\]
and consider 
\[
\mbox{av} P=\frac{1}{N} \sum_{i=1}^{N} 
PPAI_\phi(B_i).
\]
Define the piecewise constant  layered  function
\[
K({\bf x}) = 
\frac{1}{N} \sum_{i=1}^k \lambda(B_i)\frac{1}{|B_i|\ \ } 
\mbox{ for  } {\bf x}\in B_k \setminus B_{k+1},
\]
$k=1,\dots,N$, where $B_{N+1}= \emptyset$. We have that 
\begin{equation}\label{layeredcake}
\mbox{av} P=\frac{|A|}{\displaystyle \int_{A}\phi }\ \ \int_{\mathbb{R}^n}\phi ({\bf x})K({\bf x})d{\bf x}.
\end{equation}
Note that the function $K$ depends on the sets 
$B_i$ and the weight $\lambda$ but we not make this dependence explicit in our notation. 
We conclude that the average value  (of PAI indicators)  over some regions  corresponds to the inner product between $\phi$ and a piecewise constant function that weights the regions by the inverse their areas penalized by $\lambda$.
The layered  function $K$ can be taught as an approximation of a kernel that has singularities in the region
$B_N$ where it takes maximum value.

\subsection{PAI of a random variable}
Recall that $\psi$ is the density of the random variable $Y$ that we want to use to predict the hot spots of the random variable $X$. See \cite{Chainey_2008}. Assume that $\psi$ is continuous and that 
$\int_A \psi >0$. For any $s\in[0,1]$ define 
$r=r(s)$    by 

\begin{equation}\label{eq:def:r}
r(s)= \inf \Big \{r\geq 0 : \quad \int_{[\psi >r]\cap A}\psi=(1-s) \int_A \psi
\Big \}
\end{equation}
and the subregion $B_s$  by 
\begin{equation}\label{eq:def:Bs}
B_s=B_s(\psi)=[\psi>r(s)]\cap A.
\end{equation}
Define the prediction accuracy index at level $s$  as the PPAI of the subregion $B_s$.
\begin{equation}
   p(s)= p(s;\psi,\phi)=PPAI_\phi(B_s)=    \frac{ |A|}{\displaystyle \int_A \phi} \lambda(B_s)\frac{\displaystyle \int_{B_s} \phi }{
    |B_s| } .
\end{equation}
As before,   the $s$ level
    prediction accuracy index is computed by 
    dividing the hit rate by the volume percentage using the region 
    $B=[\psi>r(s)]$ and multiplying by a penalization. \\

From the comments on the previous subsection we have for the case $\lambda=1$ (no penalization), 
\begin{itemize}
    \item $ p(s;\psi,\phi)\leq \frac{R}{ \fint_A\phi}$  
     \item $\frac{r(s)}{ \fint_A\phi} \leq p(s;\phi,\phi)\leq \frac{R}{  \fint_A\phi}$     
     \item   If  $B^*=\phi^{-1}(R)\cap A$  is such that  $|B^*|>0$ then $p(1;\phi,\phi)=\frac{R}{\fint_A\phi}$. 
     \item If $B^*$ defined above is such that $|B^*|=0$ and there exists $s \in (0,1)$ such that $B_s \subset A$ with $|B_s|>0$ and $B^*\subsetneq B_s$, then for $\epsilon>0$ small enough we have $B_s\subsetneq B_{s-\epsilon}$ and
    \[
p(s-\epsilon;\phi,\phi)<p(s;\phi,\phi)\leq \frac{R}{\fint_A\phi}.
\]
  \item If $A=[\phi>0]$ (with finite volume)  we have %$p(1;\phi,\phi)=\phi({\bf x_0})$ and 
    $p(0;\phi,\phi)=1$.

\end{itemize}

As mentioned before, in order to have an overall quantity, the performance may be measured by an average of the prediction accuracy index at different levels. More precisely, chose an integer $N$ and define \begin{equation}
    P_N(\psi,\phi)=\frac{1}{N}\sum_{i=1}^Np\left(\frac{i}{N}; \psi,\phi\right).
\end{equation}

Note that, under appropriate assumptions we shall have a limiting value when $N\to \infty$. We define
\begin{equation}\label{eq:def:Ppsiphi}
    P(\psi,\phi) = \lim_{N\to \infty}  P_N(\psi,\phi)= \int_0^{1}
    p(s;\psi,\phi)ds.
\end{equation}
%We also define 
%\begin{equation}
% P(t;\psi,\phi) = \int_0^{t}p(s;\psi,\phi)ds.    
%\end{equation}

%In this paper we want to study the functions $t\mapsto P(t; \psi,\phi)$ and the mapping 
%$(\psi,\phi)\mapsto P(\psi,\phi)$.

\subsection{Average PAI and   kernels} 
Note that, by using Fubini's theorem (see also \eqref{layeredcake}), 
\begin{eqnarray*}
 P(\psi,\phi) &=& \frac{1}{ \fint_A\phi}  \int_0^{1} 
    \frac{\lambda(B_s)}{
    |B_s| } \int_{B_s} \phi({\bf y})d{\bf y} \, ds \\&=& 
    \frac{1}{ \fint_A\phi}  \int_0^{1} \displaystyle \int_{\mathbb{R}^n} \phi({\bf y})
    \frac{\lambda(B_s) }{
    |B_s| }{ 1}_{B_s}({\bf y})d{\bf y}\,  ds\\
    &=&
   \frac{1}{ \fint_A\phi}   \int_{\mathbb{R}^n}  \phi({\bf y})\left(\int_0^{1}
    \frac{\lambda(B_s) }{
    |B_s| }{ 1}_{B_s}({\bf y})ds \right)d{\bf y}=
    \frac{1}{ \fint_A\phi} \int_{\mathbb{R}^n}  \phi K_\psi d{\bf y} .
 \end{eqnarray*}
Here we have introduced the ``layered'' function 
\[
K_\psi({\bf y})=\int_0^{1}
    \frac{\lambda(B_s) }{
    |B_s| }{ 1}_{B_s}({\bf y})\, ds .
\]
In case $\psi$ is continuous then $K_\psi$ and $\psi$ have the same level sets. Indeed, if we put 
$t(\bf{y})$ defined such that $r(t(\bf{y}))=\phi({\bf y})$ where
$r(s)$ is defined in \eqref{eq:def:r} we have that 
\[
{\bf 1}_{B_s}({\bf y})=\begin{cases}
1, & s\leq t(\bf{y})\\
0, & s>t(\bf{y}).
\end{cases}
\] Then, 
\[
K_\psi({\bf y}) =
    \int_0^{t(\bf{y})}
    \frac{\lambda([\psi>r(s)]) }{
    |[\psi>r(s)]| }ds = \int_0^{t(\bf{y})}
   \frac{\lambda(B_s) }{
    |B_s| }ds.
\]
Summarizing, 
\begin{equation} \label{eq:PAIandInnerProduct}
P(\psi,\phi)=
    \frac{1}{\fint_A\phi} \int_{\mathbb{R}^n}  \phi({\bf y}) K_\psi({\bf y}) d{\bf y} .
\end{equation}
The $K_\psi$ is a positive function with the same level curves as $\psi$ with maximum value (and possible singularities) at 
$B^*=\arg\max_{\mathbb{R}^n} \psi\cap A$. Appearance of singularities depends on the integrability of the function 
\[
s\mapsto \frac{\lambda(B_s)}{|B_s|}=\frac{\lambda([\psi>r(s)])}{|[\psi > r(s)]|} \mbox{ in }
[0,1].
\]
Therefore, the value $P(\psi,\phi)=
    \frac{1}{\fint_A\phi} \int_{\mathbb{R}^n}  \phi K_\psi d{\bf x}$ will give higher weight to the regions containing  $B^*=\arg\max_{\mathbb{R}^n} \psi \,\cap A$. To further clarify our point we present the next example. 
  \begin{figure}\label{fig:ejemplo1}
\includegraphics[scale=0.15]{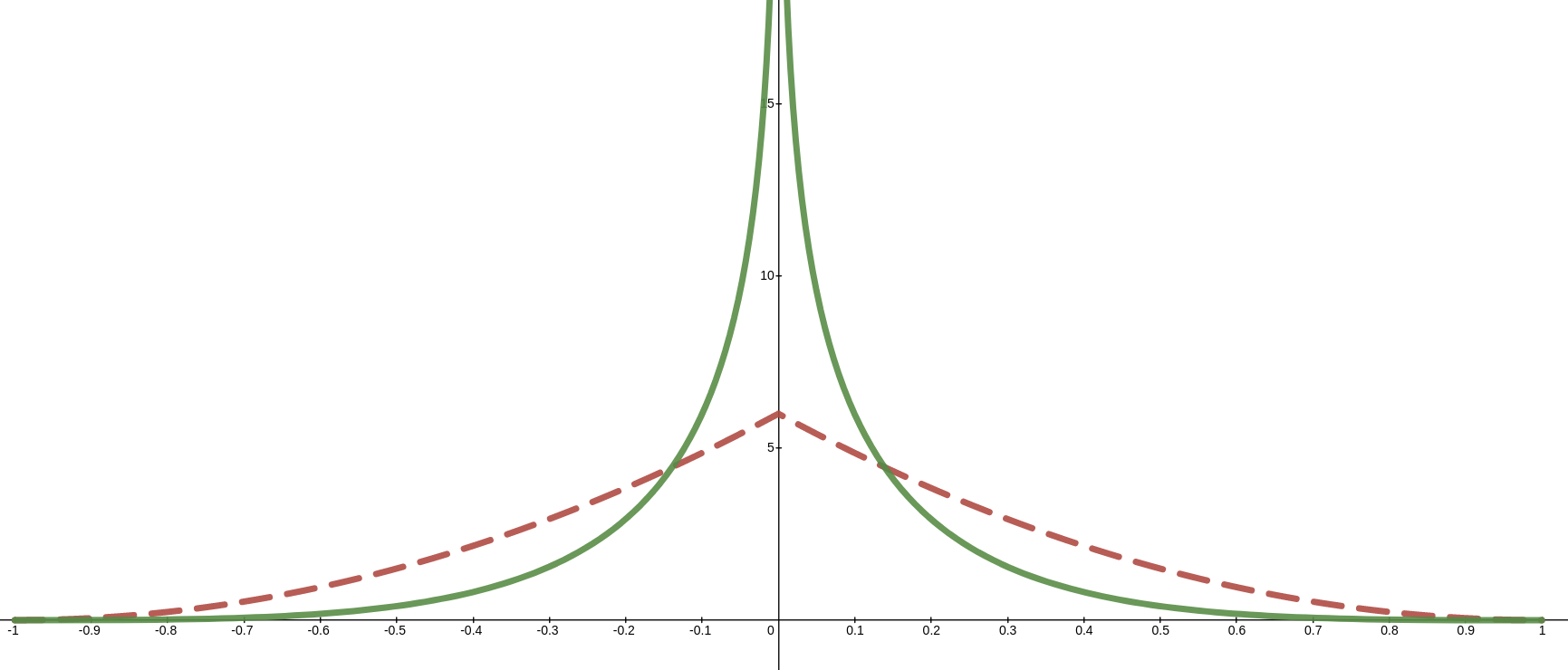}
\includegraphics[scale=0.172]{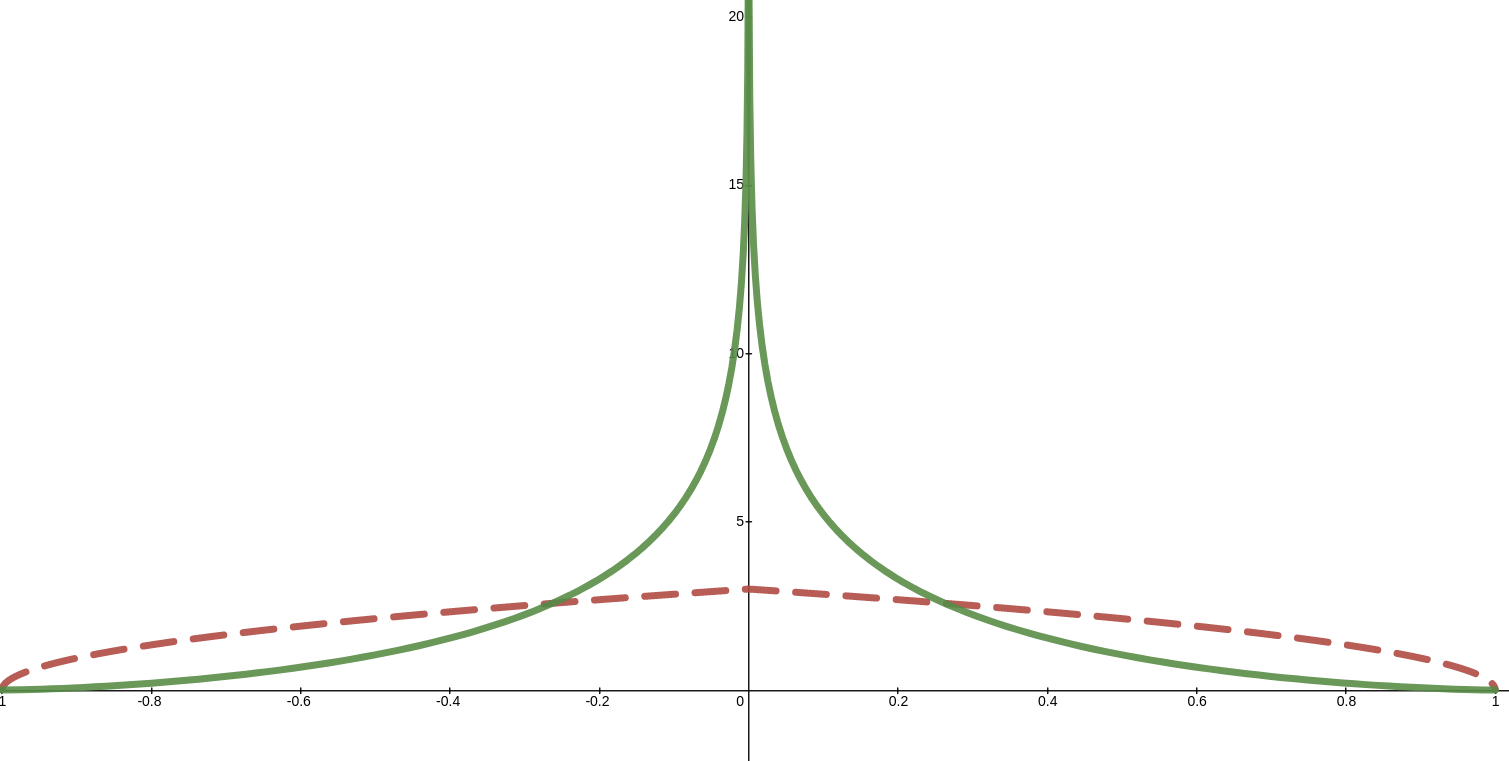}
\caption{Illustration of example 1 for $p=3$ and $p=1.5$. The function 
$\psi$ corresponds to the red dashed line and the function 
$K_\psi$ to the solid green line.
https://www.desmos.com/calculator/xcoi32o8bs .}
\end{figure}  
\begin{example}\normalfont
Consider $n=1$, $p>0$ and $\phi$ defined by 
\[
\psi(x) = \begin{cases}
0, & x<-1,\\
\frac{1}{2}p(1-|x|)^{p-1} & -1\leq x <1,\\
%\frac{1}{2}p(1-x)^{p-1} & 0\leq <1 \\
0, &1\leq x.
\end{cases}
\]
We have $\int_{-1}^1 \psi(x)=1$ and $\psi\geq 0$ and $\max_{x} \psi(x)=\frac{1}{2}p$. Additionally if we take 
$A=[-1,1]$ we have $|A|=2$ and it holds that the mass of the region 
$[\psi>r]$ is given by
\[
\int_{[\psi > r]} \psi =2\int_0^{1-(\frac{2r}{p})^\frac{1}{p-1}} \psi =  1-\left(\frac{2r}{p}\right)^\frac{p}{p-1}.
\]
Therefore, given $s$, we can find $r(s)$ such that $\int_{[\psi > r]} \psi  =1-s$
 to obtain,
 \[
r=r(s)=\frac{1}{2}ps^{\frac{p-1}{p}}.
\]
Then, the measure of the region $B_s= [\psi> r(s)]\cap A$ is the lengh of the interval 
$
\left[-1+\left(\frac{2r}{p}\right)^\frac{1}{p-1},1-\left(\frac{2r}{p}\right)^\frac{1}{p-1}\right], 
$ 
that is,
\[
|B_s|=|[\psi\geq r(s)]|=2 \left(1-\left(\frac{2r}{p}\right)^\frac{1}{p-1} \right)=
2 \left(1-(s^\frac{p-1}{p})^\frac{1}{p-1} \right)=2(1-s^\frac{1}{p}).
\]
Given $y$, we can find $t(y)$ by
\[
r(t(y))=\frac{1}{2}pt^{\frac{p-1}{p}}=\frac{1}{2}p(1-x)^{p-1}
\]
which gives 
\[
t=(1-x)^p.
\]
Thus we obtain for the case $\lambda=1$ and for $y>0$
\begin{eqnarray}
K_\psi(y)&=& \int_0^{t({y})}
   \frac{1 }{
    |B_s| }ds\\ 
&=&\int_0^{(1-y)^p} \frac{1}{2(1-s^\frac{1}{p})}ds\\
&=& \frac{1}{2} p \left(
-\log(y)+\sum_{\ell=1}^{p-1} { {p-1}\choose \ell } \frac{(-1)^\ell}{\ell} x^\ell
\right). \label{eq:Kejemplo}
\end{eqnarray}
See Figure \ref{fig:ejemplo1} for an illustration.

From \eqref{eq:PAIandInnerProduct} we have $\displaystyle P(\phi,\psi)=  2 \int_{-1}^1  \phi K_\psi d{y}$, that is, an inner product with the function $K_\psi$ in \eqref{eq:Kejemplo}. Note that $K_\psi$ puts a very high weight around the value $\phi(0)$ and therefore ignoring other possible regions with hot spots. For instance if $\{-0.5,0,0.5\}=\arg\max_{z\in \mathbb{R}} \phi $ then the possible hot spots of $\phi$ at $0.5$ and $-0.5$ wont be detected but the value 
of $P(\phi,\psi)$  will be very high (as far as we select $p$ high enough). \endproof{}
\end{example}

\begin{remark} \normalfont \label{remarkPAI}
We then conclude that not the PAI at a particular level,  not the average of several PAIs at different levels, are good measures or indicators in order to select one among different possible predictions of the random variable 
$X$, in particular, hots spots reported after selecting among several models using PAI as a main indicator are not adequate.  See \cite{joshi2021considerations,Chainey_2008}.

In order to make better hot spots predictions in practice it is recommended, before computing hots spots using PAI as indicator, that the selected model shall be decided using other more suitable measures. See for instance \cite{cha2007comprehensive,gibbs2002choosing}.
\end{remark}

\section{The integral  average transform }\label{sec:IAT}
In the previous section we presented an application where integral of averages of a function on subregions are computed. We summarize this procedure as follows:  consider a function $f: \mathbb{R}^n\to \mathbb{R}$ and  family of regions 
neighboring ${\bf x}$, say $\{B_{s,\bf x} \subset \mathbb{R}^n\}_{s>0,{\bf x}\in\mathbb{R}^n}$  such that: 
$\{{\bf x}\} \subset B_{s,\bf x}\subset B_{s+\epsilon,\bf x}$ for all $s$ and $\epsilon>0$.
Define the following integral average transform of $f$ by
\begin{equation}\label{eq:defIntegralAverageTransform}
u({\bf x})=\int_\mathbb{R} \frac{\lambda(s,{\bf x})}{|B_{s,{\bf x}}|}\int_{B_{s,{\bf x}}} f(y)dy=
\int_\mathbb{R} \lambda(s,{\bf x})\fint_{B_{s,{\bf x}}} f(y)dy.
\end{equation}
That is, a weighed integral of {\it plain} average values of $f$ over the regions $B_{s,{\bf x}}$, $s\in\mathbb{R}$.  Here $\lambda(s,{\bf x})$ is a non-negative weight function.

This procedure is equivalent to integrating against a possible singular kernel. We formalize this interpretation in the next statement. 
\begin{theorem}
Define
\[K({\bf y} , {\bf x})= 
    \int_\mathbb{R}
    \frac{\lambda(s,{\bf x}) }{|B_{s,{\bf x}}| } {\bf 1}_{B_{s, {\bf x}}}({\bf y} )ds
\]
and consider $u$ defined in  \eqref{eq:defIntegralAverageTransform}.
Then, 
$u({\bf x})=\int_{\mathbb{R}^n} f({\bf y})K({\bf y},{\bf x})dy$.
\end{theorem}

As in the previous subsections we can choose for instance, subsets related to levels curves of another functions $\psi_{\bf x}:\mathbb{R}^n\to \mathbb{R}$, say,  
\begin{equation}\label{eq:def:B_sGeneralCase}
B_{s,{\bf x}} = \psi_{\bf x}^{-1}((-\infty,s)).
\end{equation}

\begin{remark}\normalfont \label{remarkPAIandIAT}
In Section \ref{sec:motivation} the average PPAI, $P(\phi,\psi)$, defined in 
\eqref{eq:def:Ppsiphi}, was an integral average transform of 
$\phi$ where the family of sets $B_s$ were given by levels sets of $\psi$. In particular, in Section \ref{sec:motivation} $\psi$ is continuous and has a maximum value ${\bf x}_0$, then $P$ corresponds to the integral average transform of $\phi$ evaluated at  ${\bf x}_0$. The definition of subregions 
$B_s$ in \eqref{eq:def:Bs} are related to superlevel sets of 
$\psi$ while for this section, and the rest of the paper, we use $B_s$ as sublevel sets;  See (\ref{eq:def:B_sGeneralCase}). In the case of super level set the possible singularity of the associated kernel 
$K(\cdot, {\bf x} )$ will be related to the minimum value of 
$\psi_x$ that, in order to fix ideas we are assuming to be  a singleton. Recall that we are assuming 
that 
$\{{\bf x}\} \subset B_{s,\bf x}\subset B_{s+\epsilon,\bf x}$ for all $s$ and $\epsilon>0$.
\end{remark}

Let us now consider a non-negative kernel given by $K:\mathbb{R}^n\times \mathbb{R}^n\to \mathbb{R}$ a  with a possible singularity at ${\bf x}={\bf y}$ and smooth for 
 ${\bf x}\not ={\bf y}$. We can then take, for instance, 
 $\psi_{\bf x}({\bf y})= 1/K({\bf x}, {\bf y})^q$ for 
 some $q>0$. Then
\begin{equation} 
B_{s,{\bf x}} =\left[\frac{1}{s} < K(\cdot , {\bf x} )^q\right]=
\{ { \bf z} \in \mathbb{R}^n : 
s^{-1/q}<K({\bf z},{\bf x})\}.
\end{equation}
Note that $\partial B_{s,{\bf x}}$ corresponds to the $1/s$ level curve of $K(\cdot,{\bf x} )^q$.  In this case if we define 
\[
\lambda(s,{\bf x}) =\frac{1}{q \ s^{\frac{1}{q}+1}}|B_{s,{\bf x}}| 
\]
and therefore we have 
\[
  \int_\mathbb{R}
    \frac{\lambda(s,{\bf x}) }{|B_{s,{\bf x}}| } { 1}_{B_{s, {\bf x}}}({\bf y} )ds =
 \int_{\frac{1}{ K({\bf x},{\bf y} )^q}}^\infty  \frac{\lambda(s,{\bf x}) }{|B_{s,{\bf x}}| } ds
 = \int_{\frac{1}{ K({\bf x},{\bf y} )^q}}^\infty  \frac{1}{q \ s^{\frac{1}{q}+1}} ds= K({\bf x},{\bf y}).
    \]
We conclude that there is several ways to re-interpret an inner product against a kernel as a integral average transform. \\

Due to the importance of kernels and its ubiquity in pure and applied mathematics it is always useful to have several interpretations and equivalent formulations for the results that are written as  integration against kernels. In the next section we present some well know examples.

\subsection{Fundamental solution of Laplace equation as an integral average transform}

As a particular example let us consider the application of singular kernels related to the Poisson problem (\cite{gilbarg2015elliptic}). 
In order to fix ideas we do not considered the more general setting regarding minimal regularity of functions involved. Instead we assume 
that all functions are sufficiently regular in order to show the 
usage of the integral average transform \eqref{eq:defIntegralAverageTransform}.

Let $f:\mathbb{R}^n\to \mathbb{R}$. Here we need to find a function $u$ such that
\begin{equation}\label{eq:poisson}
-\Delta u = f \mbox{ in  }  \mathbb{R}^n.
\end{equation}
Consider $\psi_{\bf x}({\bf z})=||{\bf z}-{\bf x}||.$ In this case 
\begin{equation}\label{eq:defBsfGreen}
B_{s,{\bf x}}=B_s({\bf x}) = \{ { \bf z} \in \mathbb{R}^n : 
||{\bf z}-{\bf x}||< s\}.
\end{equation}
Denote $\omega_n$ the volume of the unit ball in $\mathbb{R}^n$ and define the weight
\begin{equation}\label{weight}
\lambda(s,{\bf x})=\begin{cases}
% ??, &0\leq s, \quad n=1\\\
\displaystyle 
\frac{|B_s({\bf x})|}{|\partial B_s({\bf x})|}=
\frac{\omega_n s^n}{n \omega_n s^{n-1}}= \frac{s}{n}, & 0\leq s, \quad n\geq 2,\\
0, & \mbox{otherwise}.
\end{cases}
\end{equation}
Then, in order to compute the corresponding kernel note that, 
for $R>0$ and $||{\bf x}-{\bf y}||<R$, 
\begin{eqnarray}
K_R({\bf y} , {\bf x}):&=&     \int_0^{R}
    \frac{1}{|\partial B_s({\bf x})| } {\bf 1}_{B_{s,{\bf x}}}({\bf y} )ds \\
    &=&
     \int_{||{\bf x}-{\bf y}||}^{R}
    \frac{1}{n\omega_n s^{n-1}} ds \nonumber\\
    &=&
    \begin{cases}
    \frac{1}{2\pi}\log\left(\frac{ R}{||{\bf x}-{\bf y}||}\right), & n=2,\\
    \frac{1}{n(2-n)\omega_n}\left(R^{2-n}- ||{\bf x}-{\bf y}||^{2-n}\right),& n\geq 3. 
    \end{cases} \label{GRball}\\
    &=& C_n(R) +G({\bf y} , {\bf x}) \label{GreenFunctioninRn}
\end{eqnarray}
where the constant $C_n(R)$ is given by 
\[
C_n(R)=\begin{cases}
    \frac{1}{2\pi}\log(R), & n=2,\\
    \frac{1}{n(2-n)\omega_n}R^{2-n}& n\geq 3,
    \end{cases}
    \]
and $G$ is the fundamental solution of the Laplace equation. See 
\cite{gilbarg2015elliptic,evans1998partial}.
By taking $R\to\infty$ we see that, for $n\geq 3$,
\[
K({\bf y} , {\bf x})= 
\int_0^{\infty}
    \frac{1}{|\partial B_s({\bf x})| } {\bf 1}_{B_{s,{\bf x}}}({\bf y} )ds  =G(x,y) .
\]

%Analogously,  writing $\tilde{\bf x}= \frac{R^2}{||{\bf x}||^2}{\bf x}$ for ${\bf x}\not = 0$, 
%\begin{eqnarray}
%K({\bf y} ,\tilde {\bf x})%&=& 
%\int_{||\tilde{\bf x}-{\bf y}||}^{
% {||\tilde{\bf x}||}+R}
%    \frac{1}{\mbox{Per}(B_s({\bf x})) } {\bf 1}_{B_{s,{\bf x}}}({\bf y} )ds\\
%    &=&
%     \int_{||\tilde{\bf x}-{\bf y}||}^{||\tilde{\bf x}||+R}
%    \frac{1}{n\omega_n s^{n-1}} ds\\
%    &=&\begin{cases}
%    \frac{1}{2\pi}\log\left(\frac{ R+||\tilde{\bf x}||}{||\tilde{\bf x}-{\bf y}||}\right), & n=2,\\
%    \frac{1}{n(2-n)\omega_n}\left(( R+||\tilde{\bf x}||)^{2-n}- ||\tilde{\bf x}-{\bf y}||^{2-n}\right),& n\geq 3. 
%    \end{cases}\\
%\end{eqnarray}

We have the following result.

\begin{theorem}\label{th:GreeRn}
Assume that $n\geq 3$ and that $f$ is regular enough and with compact support.
Considered the integral average transform  defined 
in \eqref{eq:defIntegralAverageTransform} with family of sets in \eqref{eq:defBsfGreen} and  
weight in \eqref{weight}, that is, 
\begin{equation} \label{eq:GreenRnSolution}
u({\bf x})= \int_{0}^{\infty}  \frac{s}{n} \fint_{||{\bf x}-{\bf y}||<s}  f({\bf y})dyds,
\end{equation}
where 
\[
\fint_{||{\bf x}-{\bf y}||<s}  f({\bf y})d({\bf y})ds
=\fint_{ B_s({\bf x}) }  f({\bf y})dyds
=\frac{1}{| B_s({\bf x})|}\int_{ B_s({\bf x})}  f({\bf y})d{\bf y}ds
\]
is the average of $f$ over the ball 
$B_s({\bf x})$. Then 
\[
-\Delta u = f.
\]
\end{theorem}
\begin{proof}
Using \eqref{GreenFunctioninRn} we have, for $R>0$, the truncated integral, 
\begin{eqnarray}
u_R({\bf x})&=&\int_{0}^{R}  \frac{s}{n} \fint_{||{\bf x}-{\bf y}||<s}  f(y)dyds=\int_{B_R({\bf 0}) }f({\bf y}) K_R({\bf x},{\bf y})d{\bf y} \\
&=&C_n(R)\int_{B_R({\bf x}) }f({\bf y})d{\bf y}+ 
\int_{{B_R({\bf x}) }}f({\bf y}) G({\bf x},{\bf y})d{\bf y}.
\end{eqnarray}
By taking $R$ large enough, we see that 
\[
u_R({\bf x})=
C_n(R)\int_{\mathbb{R}^n }f({\bf y})d{\bf y}+ 
\int_{\mathbb{R}^n}f({\bf y}) G({\bf x},{\bf y})d{\bf y}.
\]
The results follows from the fact that $G$ is the fundamental solution of the Laplace equation and the other term vanishes when $R\to \infty$ for $n\geq 3$ and $u_R({\bf x}) \to u({\bf x})$. 
\end{proof}

\begin{remark}\normalfont 
Consider now the case $n=2$, 
${\bf x}\in \mathbb{R}^n$, and 
$\epsilon >0$. For all  $R$ large enough we have
\[
u_R({\bf z})=
C_n(R)\int_{\mathbb{R}^n }f({\bf y})d{\bf y}+ 
\int_{\mathbb{R}^n}f({\bf y}) G({\bf z},{\bf y})d{\bf y}.
\]
for all ${\bf z}\in B_\epsilon({\bf x})$. 
It is enough to take $R=R_0+||x||+\epsilon$ where $B_{R_0}({\bf 0})$ contains the support of $f$.
We then see that 
$-\Delta u_R({\bf x})= f({\bf x})$.
\end{remark}

We have also the following result concerning a generalization of the mean value property. See  
\cite{delaurentis1990monte}.

\begin{theorem}
Assume that $u$ and $f$ are sufficiently smooth and satisfy
\[
-\Delta u = f \mbox{ in  } B_R({\bf x}_0) .
\]
Then the following \it{mean value property} holds,
\begin{equation}\label{eq:newmeanvalue}
u({\bf x}_0)= \fint_{\partial B_R({\bf x}_0)} u({\bf y})d{\bf y}+
\int_{0}^R  \frac{s}{n} \fint_{||{\bf x}_0-{\bf y}||<s}  f({\bf y})d{\bf y}ds.
\end{equation}

\end{theorem}
\begin{proof}
From \cite{delaurentis1990monte} we have 
\[
u({\bf x}_0)= \fint_{\partial B_R({\bf x}_0)} u({\bf y})d{\bf y}+
\int_{B_R({\bf x}_0)}   f({\bf y})G_{B_R}({\bf x}_0,{\bf y})d{\bf y}ds
\]
where $G_{B_R}$ is the Green function on the ball $B_R({\bf 0})$ given by
\[
G_{B_R}({\bf x}_0,{\bf y})=
\begin{cases}
  \frac{1}{2\pi}\log\left(\frac{ R}{||{\bf x}_0-{\bf y}||}\right), & n=2,\\
    \frac{1}{n(2-n)\omega_n}\left(R^{2-n}- ||{\bf x}_0-{\bf y}||^{2-n}\right),& n\geq 3. 
\end{cases} 
\]
From  \eqref{GRball} and the Fubini's theorem it follows
\begin{eqnarray}\nonumber
\int_{B_R({\bf x}_0)}   f({\bf y})G_{B_R}({\bf x}_0,{\bf y})d{\bf y}ds&=&
\int_{B_R({\bf x}_0)}   f({\bf y})K_R({\bf x}_0,{\bf y})d{\bf y}ds \nonumber \\
&=&
\int_{0}^R  \frac{s}{n} \fint_{||{\bf x}_0-{\bf y}||<s}  f({\bf y})d{\bf y}ds.
\end{eqnarray}
This finished the proof.
\end{proof}
We now turn to show how to rewrite the solution of the Poisson problem on the positive semi-space as an integral average transform. 
Here we need to find a function $u$ such that
\begin{equation}
-\Delta u = f \mbox{ in  }  \mathbb{R}^n_+=\{ { \bf z} \in \mathbb{R}^n : z_n>0\} 
\end{equation}
with $u=0$ on $\partial \mathbb{R}^n_+.$

\begin{theorem} 
Assume that $n\geq 3$, $f$ is regular enough and with compact support contained in  $\mathbb{R}^n_+$  and considered the integral average transform  defined by
\begin{equation} \label{eq:halfspaceAIT}
u({\bf x})= \int_{0}^{\infty}  \frac{s}{n} \fint_{B_s({\bf x})}  f(y) { 1}_{B_s({\bf x})\setminus B_s({\bf x}-2x_n{\bf e}_n) }dyds.
\end{equation} Then 
\[
-\Delta u = f\mbox{ in  }  \mathbb{R}^n_+
\]
and $u=0$ on $\partial \mathbb{R}^n_+.$
\end{theorem}
\begin{proof}
We have (see \cite{evans1998partial,gilbarg2015elliptic})
\begin{eqnarray}
u({\bf x})&=& 
\int_{\mathbb{R}^n}f({\bf y}) \left( G({\bf x},{\bf y})-
G({\bf x}-2x_n{\bf e}_n,{\bf y})d{\bf y} \right) \label{Greendiff}\\
&=& \int_{0}^{\infty}  \frac{s}{n}\left( \fint_{B_s({\bf x})}  f(y)  dy
- \fint_{B_s({\bf x}-2x_n{\bf e}_n)}  f(y)dy\right)
ds
\end{eqnarray}
Here we have used \eqref{eq:GreenRnSolution}.
The result follows from recalling that the support of $f$ is contained in 
$\mathbb{R}^n_+$ and by noting that 
\[
\fint_{B_s({\bf x})}  f(y)  dy
- \fint_{B_s({\bf x}-2x_n{\bf e}_n)}  f(y)dy=
\fint_{B_s({\bf x})}  f(y) { 1}_{B_s({\bf x})\setminus B_s({\bf x}-2x_n{\bf e}_n) }dy.
\]
\end{proof}

\begin{remark} \normalfont
For $n=2$  we have, for 
$R$ large enough 
\begin{eqnarray*}
G({\bf x},{\bf y})-
G({\bf x}-2x_n{\bf e}_n,{\bf y})d{\bf y} &=&
K_R({\bf x},{\bf y})-
K_R({\bf x}-2x_n{\bf e}_n,{\bf y})
\end{eqnarray*}
and therefore
\begin{eqnarray}
\int_{\mathbb{R}^n}f({\bf y}) \left( G({\bf x},{\bf y})-
G({\bf x}-2x_n{\bf e}_n,{\bf y})d{\bf y} \right)=\\
\int_{\mathbb{R}^n}f({\bf y}) \left( K_R({\bf x},{\bf y})-
K_R({\bf x}-2x_n{\bf e}_n,{\bf y}) \right)d{\bf y}\\
\end{eqnarray}
We conclude the same result as in
\eqref{eq:halfspaceAIT} by taking 
$R\to\infty$.
\end{remark}

\begin{remark} \normalfont \label{remarkGreenIAT}
In the previous results we write the solution of the Poisson equation in the half space as an integral of averages of the forcing term over balls centered at 
${\bf x}$. In the case the ball 
$B_s({\bf x})\not  \subset \mathbb{R}^n_+ $, the function $f$ is cut to 
$f { 1}_{B_s({\bf x})\setminus B_s({\bf x}-2x_n{\bf e}_n) }$. See Figure \ref{fig:halfplane}.

\end{remark}

\begin{figure}
\begin{tikzpicture}[scale=1.2]

\filldraw[color=red!60, fill=blue!5, very thick](0,1) circle (1.5);  
\filldraw[black] (0,1) circle (2pt) node[anchor=west]{${\bf x}$};
\filldraw[color=red!60, fill=red!0, very thick](0,-1) circle (1.5);  
\filldraw[black] (0,-1) circle (2pt) node[anchor=west]{${\bf x}-2x_n{\bf e}_n$};
 \draw[gray, thick] (-4,0) -- (4,0);
%  \draw[gray, thick] (0,-2) -- (0,2);
\end{tikzpicture}
\caption{Illustration of $B_s({\bf x})\setminus B_s({\bf x}-2x_n{\bf e}_n)$. Formula \eqref{eq:halfspaceAIT} } illustrate that for the average of $f$ on $B_s({\bf x})$ only values of the shaded region, $B_s({\bf x})\setminus B_s({\bf x}-2x_n{\bf e}_n)$, are considered in the integration.
\label{fig:halfplane}
\end{figure}
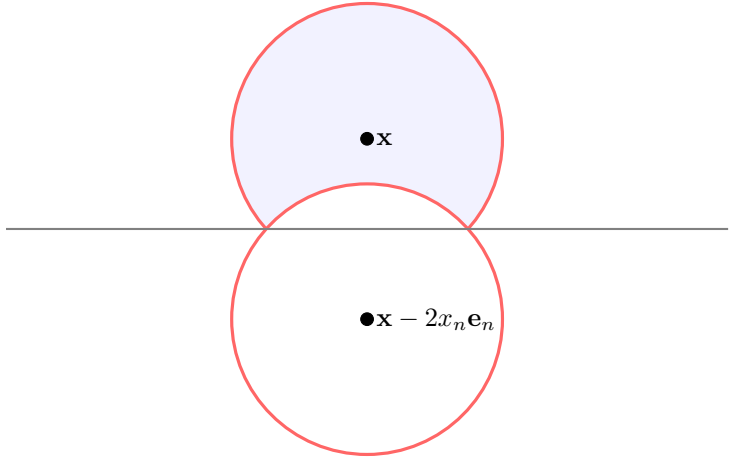

Finally we mention that formula \eqref{eq:GreenRnSolution} give us additional interpretation of the solution of the Poisson equations in domain with boundaries. Instead of subtracting two solution as in formula \eqref{Greendiff}, we could extend the forcing term in such a way that averages centered at the boundary vanish. For instance, 
let us consider the domain $\mathbb{R}^n_+$ and ${\bf x}\in \partial \mathbb{R}^n_+$. In order to use \eqref{eq:GreenRnSolution} to obtain the solution of 
the Poisson equation in $\mathbb{R}^n_+$  we could extend $f$ to 
the whole $\mathbb{R}^n$, say $E(f):\mathbb{R}^n\to \mathbb{R}$ such that $E(f)|_{\mathbb{R}^n_+} = f$ and with  
$\int_{B_s({\bf x})}f({\bf y})d{\bf y}=0$ for all 
 ${\bf x}\in \partial \mathbb{R}^n_+$. We then have the following 
results as an alternative the Theorem 3.4.
\begin{theorem} \label{greeextension}
Assume that $n\geq 3$, $f$ is regular enough and with compact support contained in  $\mathbb{R}^n_+$  and considered 
\[
Ef({\bf y})=
\begin{cases}
  f({\bf y}),  & {\bf y}\in \mathbb{R}^n_+,\\
  0, &{\bf y}\in \partial\mathbb{R}^n_+,\\
  -f({\bf y}-2y_n{\bf e}_n),  & \mbox{elsewhere},\\
\end{cases} 
\]
and $u$ defined  by
\begin{equation} \label{eq:halfspaceAITextension}
u({\bf x})= \int_{0}^{\infty}  \frac{s}{n} \fint_{B_s({\bf x})}  Ef(y) dyds.
\end{equation} Then 
\[
-\Delta u = f\mbox{ in  }  \mathbb{R}^n_+
\]
and $u=0$ on $\partial \mathbb{R}^n_+.$
\end{theorem}
\begin{proof}
Observe that formula \eqref{eq:halfspaceAITextension}
coincides with \eqref{eq:halfspaceAIT}.
\end{proof}
Similar argument can be applied to other domains with simple boundaries such as strips and cubes. 

\section{Final comments}\label{sec:conclusions}
In this short note we introduced the integral average transform. Our motivation comes from the need of a better mathematical understanding of some practical measures such as the prediction accuracy index that is popular in problems related to predictive security. In this paper we have explained the mathematical and practical context of this application in order to motivate
our main definition. The integral average transform is defined for a given function 
$f({\bf y})$ and it is the result of integrating plain averages of $f$ over an family of sets (containing a point ${\bf x}$) indexed by the integration argument. See 
\eqref{eq:defIntegralAverageTransform} for a precise definition. We show that this procedure is equivalent to computing integrals against a two variable kernel $K({\bf y},{\bf x})$ with a possible singularity at ${\bf x}$. We also show that 
any kernel integral of the form $\int f({\bf y})K({\bf y},{\bf x})d{\bf y}$ can also be interpreted as an integral average transform. Given the ubiquity of   kernels $K({\bf y},{\bf x})$ in the solution of many problems, we believe that this novel interpretation may be worth of further investigation. For instance, other kernels can be considered such as Poisson's kernels. We can also associate some dynamics to the parameter $s$.

\bibliographystyle{plain}
\bibliography{references}

\end{document}